\documentclass[journal]{IEEEtran}
\usepackage{cite}

\ifCLASSINFOpdf
\else
\fi
\usepackage{tikz}

\usepackage[cmex10]{amsmath}
\usepackage{amssymb}
\usepackage{amsthm}

\newtheorem{theorem}{Theorem}[section]
\newtheorem{proposition}[theorem]{Proposition}

\newcommand{\vect}[1]{\mathbf{#1}} 

\newcommand{\bfx}{\vect{x}}
\newcommand{\bfy}{\vect{y}}
\newcommand{\bfz}{\vect{z}}

\newcommand{\bfY}{\vect{Y}}

\newcommand{\supp}{\mathfrak{S}}
\newcommand{\abs}[1]{\lvert#1\rvert} 
\newcommand{\card}[1]{\abs{#1}} 
\newcommand{\comp}[1]{{#1}^{\textnormal{c}}} 
\newcommand{\eqdef}{\triangleq} 
\newcommand{\set}[1]{\mathcal{#1}} 
\newcommand{\zeronorm}[1]{\left\|#1\right\|_{0}}

\newcommand{\indeg}{\textnormal{d}_{\textnormal{in}}}
\newcommand{\outdeg}{\textnormal{d}_{\textnormal{out}}}

\newcommand{\Cal}{C_{\textnormal{l}}}
\newcommand{\Csp}{C_{\textnormal{Sp}}}

\newcommand{\Ceo}{C_{\textnormal{0-u}}}
\newcommand{\Ceof}{C_{\textnormal{0-u,fb}}}

\begin{document}
%
\title{The Zero-Undetected-Error Capacity Approaches the Sperner Capacity}
%
%
%

\author{Christoph~Bunte,
        Amos~Lapidoth,~\IEEEmembership{Fellow,~IEEE,}
        and~Alex~Samorodnitsky
\thanks{C. Bunte and A. Lapidoth are with the Signal and Information Processing
Laboratory at ETH Zurich. E-mail: \{bunte,lapidoth@isi.ee.ethz.ch\}.}
\thanks{A. Samorodnitsky is with the Institute of Compute Science at The Hebrew
University of Jerusalem. His work is partially supported by grants from BSF and
ISF. E-mail: salex@cs.huji.ac.il}
\thanks{This paper was presented in part at the 2014 ISIT conference.}
}

\maketitle

\begin{abstract}
Ahlswede, Cai, and Zhang proved that, in the noise-free limit, 
the zero-undetected-error capacity is lower-bounded by the 
Sperner capacity of the channel graph, 
and they conjectured equality. Here we derive an upper bound that proves the conjecture.
\end{abstract}


%
\IEEEpeerreviewmaketitle

\section{Introduction}
\label{sec:intro}
\IEEEPARstart{A}{zero}-undetected-error decoder (z.u.e.\ decoder) declares that a message was transmitted only if
it is the only message that could have produced the observed output. If the
output could have been produced by two or more messages, it declares an erasure.
Such a decoder thus never errs: it either produces the correct message or an
erasure. 

The zero-undetected-error capacity (z.u.e.\ capacity) $\Ceo$ of a channel is the supremum of all
rates that are achievable with a z.u.e.\ decoder in the sense that
the probability of erasure tends to zero as the blocklength tends to
infinity~\cite{csiszar1995channel, ahlswede1996erasure}. (It does
not matter whether we define~$\Ceo$ using an average or a maximal erasure
probability criterion.)
Restricting the decoding rule cannot help, so $\Ceo$ never exceeds the
Shannon capacity~$C$. 

Although partial results exist (see Section~\ref{sec:background}), 
the z.u.e.\ capacity
of general discrete memoryless channels (DMCs) is still unknown. 
The focus of this paper is the z.u.e.\ capacity of nearly
noise-free channels. More precisely, we focus on 
$\varepsilon$-noise channels, that is, DMCs whose input alphabet $\set{X}$ is a
subset of their output alphabet $\set{Y}$ and whose transition law $W$ satisfies
\begin{equation}
\label{eq:low_noise}
W(x|x)\geq 1-\varepsilon\quad \text{for all $x\in \set{X}$.}
\end{equation}
Here and throughout we assume that $0\leq \varepsilon<1$. 
For $\varepsilon$-noise channels we derive an upper bound on~$\Ceo$. We
then apply this result to study the limit of $\Ceo$ as $\varepsilon$ tends to zero. 
Ahlswede, Cai, and Zhang proved
that this limit is lower-bounded by the Sperner capacity of a
certain related graph, and they conjectured equality~\cite{ahlswede1996erasure}. 
Our upper bound proves this conjecture. 

The Sperner capacity is defined using graph-theoretic language
in Section~\ref{sec:graphs}. 
Here we give an alternative characterization in terms of codes (see
also~\cite{ahlswede1996erasure}). 
For this we need some standard notation. 

A DMC is specified by its transition law $W(y|x)$, $x\in \set{X}$, $y\in \set{Y}$, where $\set{X}$
and $\set{Y}$ are finite input and output alphabets. 
Feeding a sequence of input symbols $\bfx=(x^{(1)},\ldots,x^{(n)})$ to a DMC of
transition law $W$ produces a random sequence of output
symbols $\bfY=(Y^{(1)},\ldots,Y^{(n)})$ whose joint probability mass function 
(PMF) is
\begin{equation}
W^n(\bfy|\bfx) \eqdef \prod_{1\leq j\leq n} W(y^{(j)}|x^{(j)}),\quad \bfy \in \set{Y}^n.
\end{equation}
The support of $W$ is the set of all pairs $(x,y) \in \set{X}\times
\set{Y}$ for which $W(y|x)$ is positive; it is denoted by $\supp(W)$. 
Similarly, if $P$ is a PMF on $\set{X}$, then $\supp(P)$ 
denotes the set of all $x\in \set{X}$ for which $P(x)$ is positive.
We write~$PW$ for the PMF on $\set{Y}$ induced by $P$ and the channel $W$ 
\begin{equation}
(PW)(y) = \sum_{x\in \set{X}} P(x) W(y|x),\quad y \in \set{Y}.
\end{equation}
If $A\subseteq \set{X}$, then we write $P(A)$ in lieu of $\sum_{x\in A} P(x)$. 
The Cartesian product of two sets~$A$ and~$B$ is denoted by~$A\times B$. 
The $n$-fold Cartesian product of $A$ with itself is denoted by $A^n$, and 
the cardinality of $A$ is denoted by $\card{A}$. All logarithms are
natural logarithms, and we adopt the convention $0\log \frac{1}{0} = 0$. 

We define a blocklength-$n$ Sperner code for a DMC $W$ with $\set{X}\subseteq \set{Y}$ and $W(x|x)>0$ for all $x\in
\set{X}$ as a collection of length-$n$ codewords $\bfx_1,\ldots,\bfx_M$ with the property
\begin{equation}
\label{eq:zero_error_detection}
W^n(\bfx_m|\bfx_{m'}) =0\quad \text{whenever $m\neq m'$}.
\end{equation}
The rate of the code is $n^{-1} \log M$. The largest rate of a Sperner code
is a function of the channel law~$W$ and the blocklength~$n$. In fact, it 
depends on~$W$ only via its support~$\supp(W)$. 
The supremum over $n$ of the largest rate of blocklength-$n$ Sperner codes
is the Sperner capacity~$\Csp$ of the channel. 

We emphasize that a Sperner code is not a zero-error code. Indeed,
the definition does not preclude the existence of a sequence $\bfy \in
\set{Y}^n$ such that $W^n(\bfy|\bfx_{m})>0$ and $W^n(\bfy|\bfx_{m'})>0$ for some
$m\neq m'$. What cannot happen with a Sperner code is that one 
codeword is corrupted by the channel into another codeword.
In other words, at the output we either receive the correct codeword or some
sequence of output letters that is not a codeword. 

In the notation above, Ahlswede, Cai, and Zhang proved in
\cite[Theorem 2]{ahlswede1996erasure} that for $\varepsilon$-noise channels,
\begin{equation}
\label{eq:ahlswede_lb_221}
\varliminf_{\varepsilon\to 0} \Ceo \geq \Csp,
\end{equation}
and they conjectured equality.
Our main result is the following inequality.
\begin{theorem}
\label{thm:upper_bound}
For every $\varepsilon$-noise channel,
\begin{equation}
\label{eq:upper_bound}
\Ceo \leq \log\bigl(e^{\Csp}+\varepsilon
\card{\set{X}}(\card{\set{Y}}-1)\bigr).
\end{equation}
\end{theorem}
Combining Theorem~\ref{thm:upper_bound} with~\eqref{eq:ahlswede_lb_221} 
proves their conjecture:
\begin{theorem}
\label{thm:limit}
For $\varepsilon$-noise channels,
\begin{equation}
\label{eq:eq}
\lim_{\varepsilon \to 0} \Ceo = \Csp,
\end{equation}
where the limit is to be understood in a uniform sense with respect to all
$\varepsilon$-noise channels with given $\supp(W)$. 
\end{theorem}
A proof of Theorem~\ref{thm:upper_bound} is given in Section~\ref{sec:proof}. 
Before providing an outline of this proof,
we try to explain why Corollary~\ref{thm:limit} is plausible. If we use a Sperner code 
in conjunction with a z.u.e.\ decoder, then an erasure can occur
only if the codeword is corrupted, which happens with probability at
most~$1-(1-\varepsilon)^n$. 
This suggests that~$\Csp$ should be a lower bound to~$\Ceo$
when~$\varepsilon$ is very small (ignoring the issue that~$n$ tends to infinity
before~$\varepsilon$ tends to zero). 
Conversely, any code whose maximal probability of erasure under z.u.e.\ decoding is
smaller than~$(1-\varepsilon)^n$ must be a Sperner code. Since for all rates
strictly smaller than $\Ceo$ the
probability of erasure can be driven to zero exponentially fast (see below),
this suggests that $\Csp$
should be an upper bound on $\Ceo$ for small
$\varepsilon$ (ignoring the issue that the exponent of the erasure
probability may become arbitrarily small as $\varepsilon$ becomes small and the
rate approaches~$\Ceo$)

As to the outline of the proof of Theorem~\ref{thm:upper_bound}, we first
show that a multi-letter version of Forney's lower bound on $\Ceo$ is
asymptotically tight even when the input distributions are restricted to be
uniform over their support (Section~\ref{sec:forney}). We then upper-bound
the multi-letter expression using Jensen's inequality followed by algebraic
manipulations that yield a still looser bound. Thanks to the input distribution
being uniform, this looser bound depends
only on $\varepsilon$ and the support of $W$. 
The final step is to use graph-theoretic techniques, which are introduced in
Section~\ref{sec:graphs}, to obtain the desired upper bound. These techniques include upper-bounding a sum
that depends only on the in-degrees of the vertices of a graph $G$ by the
maximum size of any induced acyclic subgraph of $G$. They also include
showing that the Sperner capacity of a graph $G$ can be expressed as the
limit as $n$ tends to infinity of $1/n$ times the logarithm of the maximum
cardinality of any induced acyclic subgraph of the $n$-fold strong product of
$G$ with itself.

\section{Background}
\label{sec:background}
To put our result into perspective, we briefly review some of the
literature on the z.u.e.\ capacity and related concepts. 
We begin with Forney's lower bound
\begin{equation}
\label{eq:forney_lb}
\Ceo \geq \max_{P} \sum_{y\in \set{Y}} (PW)(y) \log \frac{1}{P(\set{X}(y))},
\end{equation}
where the maximum is over all PMFs on the input alphabet $\set{X}$, and where $\set{X}(y)$
denotes the set of all $x\in \set{X}$ for which $W(y|x)$ is positive~\cite{forney1968exponential}. 
His bound can be derived using standard random coding where each component of each codeword is drawn IID from
a PMF $P$. In general, \eqref{eq:forney_lb} is not tight.\footnote{An example where it is
not tight is the Z-channel \cite{telatar1992phd}.}
A tighter lower bound was derived in \cite{telatar1992phd, ahlswede1996erasure,
csiszar1995channel} using random coding over constant composition codes
\begin{equation}
\label{eq:hui_lb}
\Ceo \geq \max_{P}\min_{\substack{V\ll W:\\PV=PW}}I(P,V).
\end{equation}
The minimum in~\eqref{eq:hui_lb} is over all auxiliary DMCs $V(y|x), x\in \set{X}, y \in \set{Y}$ such that $V(y|x)=0$ whenever
$W(y|x)=0$ (in short $V\ll W)$ and such that $V$ induces the same output
distribution under $P$ as the true channel $W$. 

Since any code for the product channel $W^n$ is also a code for the channel $W$
of $n$ times the blocklength and $1/n$ times the rate, it follows that the
bounds~\eqref{eq:forney_lb} and~\eqref{eq:hui_lb} can be improved by applying
them to $W^n$ and normalizing the result by~$1/n$. For example, the $n$-letter
version of~\eqref{eq:forney_lb} is
\begin{equation}
\label{eq:multi_forney}
\Ceo \geq \frac{1}{n} \max_{P} \sum_{\bfy \in \set{Y}^n} (PW^n)(\bfy) \log
\frac{1}{P(\set{X}^n(\bfy))},
\end{equation}
where the maximum is over all PMFs $P$ on $\set{X}^n$, and where $\set{X}^n(\bfy)$ denotes the 
set of all $\bfx \in \set{X}^n$ for which~$W^n(\bfy|\bfx)>0$. 
A numerical evaluation in~\cite{ahlswede1996erasure} 
of the single-letter and two-letter versions
of~\eqref{eq:hui_lb} for a particular channel 
suggests that a strict improvement is possible and hence that~\eqref{eq:hui_lb}
is not always tight. 

The $n$-letter version of~\eqref{eq:hui_lb} becomes tight as $n$ tends to
infinity~\cite{csiszar1995channel, ahlswede1996erasure}. Since the proof
of~\eqref{eq:hui_lb} shows that for all rates less than the RHS
of~\eqref{eq:hui_lb} the probability of erasure can be driven to zero
exponentially fast, it follows that this is also true for all rates less than the
$n$-letter version of the bound, and hence for all rates less than~$\Ceo$. 

In Section~\ref{sec:forney} we prove that also the $n$-letter version of the weaker
bound~\eqref{eq:forney_lb} is asymptotically tight, and that this is true 
even when the maximization is restricted to PMFs that are uniform
over their support. This result will be crucial in the proof of
Theorem~\ref{thm:upper_bound}.
 
We have already pointed out that the Shannon capacity $C$ is an upper bound to
$\Ceo$. 
In fact, this bound is often tight. 
Indeed, Pinsker and Sheverdyaev~\cite{pinsker1970transmission} proved that $\Ceo=C$ whenever the
bipartite channel graph is acyclic. 
The bipartite channel graph is the undirected bipartite graph whose two independent sets of vertices are the input and output
alphabets of the channel, and where there is an edge between an 
input~$x$ and an output~$y$ if~$W(y|x)>0$ (it is customary 
to draw the inputs on the left and the outputs on the right and to label the
edges with the transition probabilities). 
Acyclic means that we cannot find an integer~$n\geq
2$, distinct inputs $x_1,\ldots,x_n$ and distinct outputs
$y_1,\ldots,y_{n}$ such that $W(y_j|x_j)>0$ and $W(y_j|x_{j+1})>0$ for
all~$j\in\{1,\ldots,n\}$ where $x_{n+1} =x_1$. Important examples of channels with
acyclic bipartite channel graphs are the binary erasure channel and the
Z-channel.

Csisz\'ar and Narayan~\cite{csiszar1995channel} proved that $\Ceo=C$ whenever there exist 
positive functions~$A$ and~$B$ such that
\begin{equation}
\label{eq:factorization}
W(y|x) = A(x)B(y),\quad \text{whenever $W(y|x)>0$.}
\end{equation}
They also proved that in this case~\eqref{eq:hui_lb} is tight. Their result is
stronger than Pinsker-Sheverdyaev because
every DMC with an acyclic bipartite channel graph possesses a factorization of its
channel law in the sense of~\eqref{eq:factorization}, but the converse is
not true. For example, the graph shown in
Figure~\ref{fig:factorize_example}
contains a cycle, yet the channel law factorizes with the choice
$A=B=1/\sqrt{2}$. 
\begin{figure}
\centering
\begin{tikzpicture}[scale=2]
\draw[-] (0,0) node[left] {$2$} -- node[above] {$1/2$} (2,0) node[right] {$2$};
\draw[-] (0,1) node[left] {$1$} -- node[above] {$1/2$} (2,1) node[right] {$1$};
\draw[-] (0,2) node[left] {$0$} -- node[above] {$1/2$} (2,2) node[right] {$0$};
\draw[-] (0,0) -- node[above] {$1/2$} (2,1);
\draw[-] (0,1) -- node[above] {$1/2$} (2,2);
\draw[-] (0,2) -- node[above, near start] {$1/2$} (2,0);
\end{tikzpicture}
\caption{The graph contains a cycle but the channel law factorizes.}
\label{fig:factorize_example}
\end{figure}
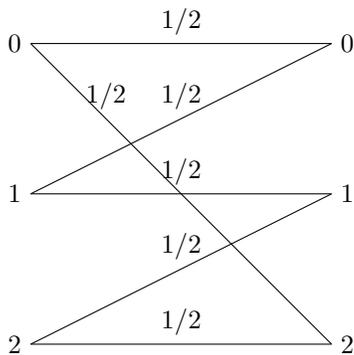
In fact, they conjecture that a necessary condition for~$\Ceo=C$ is 
that~\eqref{eq:factorization} hold on some capacity-achieving
subset of inputs (which is clearly also sufficient). 

The fact that $\Ceo=C$ for the Z-channel can be used to characterize channels
with positive z.u.e.\ capacity. Indeed, suppose there exist two input
symbols $x$ and $x'$ and an output symbol~$y$ such 
that~$W(y|x)>0$ and~$W(y|x')=0$. By combining all output symbols other than $y$ into a single
distinct output symbol, and by using only the inputs $x$ and $x'$, we can reduce the channel
to a Z-channel with crossover probability $1-W(y|x)$. For this channel~$C>0$,
and hence also $\Ceo>0$. Conversely, if $W(y|x)>0$ implies that $W(y|x')>0$ for
all~$x' \in \set{X}$, then any received sequence of output symbols can be produced by
every codeword and the decoder must always erase, so $\Ceo$ must be zero. 

We also mention that z.u.e.\ capacity is a special case of $d$-capacity 
\cite{csiszar1995channel}, and the lower bound~\eqref{eq:hui_lb} is a special
case of a lower bound on $d$-capacity~\cite{hui1983phd, csiszar1981graph,
csiszar1995channel}. It
was proved in~\cite{balakirsky1995converse} that this lower bound on $d$-capacity is tight
for binary-input channels. The proof is complicated and is, in fact, not needed
when one is interested only in z.u.e.\ capacity. Indeed, the binary-input case
is easily solved using the Pinsker-Sheverdyaev result and the following
proposition, which can also be used to improve on~\eqref{eq:upper_bound} when
$\card{\set{Y}}$ is larger than $\card{\set{X}}+2^{\card{\set{X}}}-1$ (see
Section~\ref{sec:remarks}).
\begin{proposition}
\label{prop:combine}
Suppose that the output symbols $y,y'\in \set{Y}$ are such that for every $x \in \set{X}$,
\begin{equation}
W(y|x)>0 \iff W(y'|x)>0.
\end{equation}
Then the z.u.e.\ capacity is unaltered when we combine 
$y$ and $y'$ into a single output symbol distinct from all other
output symbols. 
\end{proposition}
\begin{proof}
The set of messages that cannot be ruled out when $\bfy$ is observed at the
output is unchanged when any occurrence of $y$ in $\bfy$ is replaced with $y'$ or vice versa.
\end{proof}
Using Proposition~\ref{prop:combine} we can reduce the output alphabet of a
given DMC $W$ with input alphabet $\set{X}$ to at most $2^{\card{\set{X}}}-1$ symbols without changing
the z.u.e.\ capacity (there are $2^{\card{\set{X}}}-1$ possible combinations of
inputs that can connect to a given output). In particular, when the input
alphabet is binary, we can reduce the channel to an asymmetric binary erasure channel 
(possibly with some transition probabilities equal to zero). Since this channel
has an acyclic bipartite channel graph, 
computing the z.u.e.\ capacity of a binary-input channel can thus be reduced to
computing the Shannon capacity of an asymmetric binary erasure channel. 

Unlike the Shannon capacity, the z.u.e.\ capacity can be increased by
feedback~\cite{bunte2012zero}. In fact, with feedback the z.u.e.\ capacity
$\Ceof$ is~\cite{nakiboglu2012errors, bunte2012zero}
\begin{equation}
\Ceof = \begin{cases} C& \text{if $\Ceo>0$,}\\0& \text{if $\Ceo=0$.}\end{cases}
\end{equation}

A concept closely related to the z.u.e.\ capacity is the
listsize capacity~\cite{ahlswede1996erasure, telatar1997zero}: 
Suppose that instead of erasing when more than one message could have produced 
the observed output, the decoder instead produces a list of all messages that it cannot 
rule out. The listsize capacity is the supremum of all rates that are achievable
in the sense that the $\rho$-th moment of the size of the list produced by the
decoder tends to one as the blocklength tends to infinity; 
it is denoted by $\Cal(\rho)$. Here~$\rho$ can be any positive number; 
the case $\rho=1$ has been called average-listsize capacity~\cite{ahlswede1996erasure}.
It is not difficult to show that $\Cal(\rho)$ never exceeds $\Ceo$ 
and that $\Cal(\rho)>0$ if, and only if,
$\Ceo>0$~\cite{telatar1997zero,bunte2013listsize}. 
Lower bounds on~$\Cal(\rho)$ analogous to~\eqref{eq:forney_lb}
and~\eqref{eq:hui_lb} can be found in~\cite{forney1968exponential,
ahlswede1996erasure, telatar1997zero, bunte2013listsize}. Results for channels with feedback can be
found in~\cite{bunte2013cutoff,bunte2013listsize}, where, in particular, it is shown that feedback
can increase the listsize capacity. 

It was proved in~\cite{ahlswede1996erasure} that Corollary~\ref{thm:limit} is
true when $\Ceo$ is replaced with $\Cal(1)$
\begin{equation}
\label{eq:list_sperner}
\lim_{\varepsilon\to 0} \Cal(\rho)\Big\vert_{\rho=1} = \Csp.
\end{equation}
This, in fact, holds for all $\rho>0$. Indeed, Theorem~\ref{thm:upper_bound} 
and the fact that $\Cal(\rho)$ is upper-bounded by $\Ceo$ for all $\rho>0$
imply that
\begin{equation}
\varlimsup_{\varepsilon \to 0} \Cal(\rho) \leq \Csp.
\end{equation}
The reverse inequality for $0<\rho <1$ follows from~\eqref{eq:list_sperner} and the
fact that $\Cal(\rho)$ is nonincreasing in $\rho$. The case~$\rho
\geq 1$ is not difficult to obtain from a generalization to all $\rho>0$ of 
Forney's lower bound on $\Cal(1)$~\cite{forney1968exponential,bunte2013listsize}.

It is shown in~\cite{bunte2013listsize} that if the channel law factorizes in the sense of~\eqref{eq:factorization}, then 
\begin{equation}
\label{eq:list_cutoff}
\Cal(\rho) = \max_{P} \frac{E_0(\rho,P)}{\rho},
\end{equation}
where $E_0(\rho,P)$ is Gallager's function (see~\cite{gallager1968information}).
The RHS of~\eqref{eq:list_cutoff} is also known as the cutoff
rate~\cite{arikan1996inequality}. In fact, the listsize capacity relates to the
cutoff rate in much the same way that the z.u.e.\ capacity relates
to the Shannon capacity~\cite{bunte2013cutoff,bunte2013listsize}. 

\section{A Multi-Letter Formula for $\Ceo$} 
\label{sec:forney}
In this section we show that~\eqref{eq:multi_forney}
is asymptotically tight even when the input PMFs are restricted to be uniform
over their support. 
\begin{theorem}
For any DMC, 
\begin{equation}
\label{eq:multiletter}
\Ceo = \lim_{n\to\infty} \frac{1}{n} \max_{P\in U_n} \sum_{\bfy \in \set{Y}^n } (PW^n)(\bfy) \log
\frac{1}{P(\set{X}^n(\bfy))},
\end{equation}
where $U_n$ denotes the collection of PMFs on $\set{X}^n$ that are uniform over their
support. Moreover, the limit is equal to the supremum.
\end{theorem}
\begin{proof}
It is straightforward to verify that the sequence on the RHS
of~\eqref{eq:multiletter} without the $1/n$ factor is superadditive, which 
implies that the limit is equal to the supremum.\footnote{A
sequence $a_1,a_2,\ldots$ of real numbers is superadditive if $a_{n+m} \geq
a_n+a_m$ for every $n$ and $m$. For superadditive sequences $a_n/n$ tends to
$\sup_n a_n/n$ \cite[Problem 98]{polya1998problems}.}  
Let us denote this limit by $\lambda$. 
Achievability, i.e., $\Ceo \geq \lambda$, follows
because~\eqref{eq:multi_forney} holds for every $n$. As to the
converse, let $\bfx_1,\ldots,\bfx_{M}$ be a codebook of  blocklength $n$ and rate $R$ 
with maximal probability of erasure
under z.u.e.\ decoding less than some~$\delta\in(0,1)$
\begin{equation}
\label{eq:max_prob_of_erasure}
\max_{1\leq m \leq M} \sum_{\substack{\bfy\in \set{Y}^n:\\M(\bfy)>1}}
W^n(\bfy|\bfx_m)< \delta,
\end{equation}
where $M(\bfy)$ denotes the number of messages that cannot be ruled out when $\bfy$
is observed at the output
\begin{equation}
\label{eq:list}
M(\bfy) = \bigl|\bigl\{1\leq m \leq M: W^n(\bfy|\bfx_m)>0\bigr\}\bigr|.
\end{equation}
Condition~\eqref{eq:max_prob_of_erasure} implies that $\bfx_{m} \neq \bfx_{m'}$
when $m\neq m'$ because otherwise, as we next argue, the conditional probability
of erasure given that the $m$-th message was sent would be one. 
Indeed, if $\bfx_m = \bfx_{m'}$ for some~$m\neq
m'$, then $M(\bfy) \geq 2$ whenever $W^n(\bfy|\bfx_m)>0$
(because then also $W^n(\bfy|\bfx_{m'})>0$), and
hence 
\begin{equation}
\sum_{\substack{\bfy\in \set{Y}^n:\\M(\bfy)>1}}
W^n(\bfy|\bfx_m) = 1.
\end{equation}
Having established that the codewords are distinct, we choose $P$ to be the uniform PMF on the codebook. 
Then $P \in U_n$ and 
\begin{equation}
\label{eq:206}
P\bigl(\set{X}^n(\bfy)\bigr) = \frac{M(\bfy)}{M},\quad \text{for all
$\bfy\in \set{Y}^n$.}
\end{equation}
We further observe that 
\begin{align}
\lambda&\geq\frac{1}{n} \sum_{\bfy \in \set{Y}^n}  (P W^n)(\bfy) \log
\frac{1}{P(\set{X}^n(\bfy))}\label{eq:212}\\
&=R- \frac{1}{n} \sum_{\substack{\bfy\in \set{Y}^n:\\ M(\bfy) >1}}  (P
W^n)(\bfy) \log M(\bfy)\label{eq:214}\\
&\geq R\Biggl(1- \sum_{\substack{\bfy\in \set{Y}^n:\\ M(\bfy) >1}}  (P
W^n)(\bfy)\Biggr)\label{eq:214_2}\\
&=R\Biggl(1- \frac{1}{M}\sum_{m=1}^{M} \sum_{\substack{\bfy\in \set{Y}^n:\\M(\bfy)>1}}
W^n(\bfy|\bfx_m)\Biggr)\label{eq:214_3}\\
&> R(1-\delta),\label{eq:215}
\end{align}
where \eqref{eq:212} follows because $\lambda$ is the supremum of a sequence
whose $n$-th term is no smaller than the RHS of~\eqref{eq:212};
where~\eqref{eq:214} follows from~\eqref{eq:206} and the fact that $\log 1 =0$;
where~\eqref{eq:214_2} follows because $M(\bfy) \leq M$;
where~\eqref{eq:214_3} follows from the choice of $P$; 
and where~\eqref{eq:215} follows from~\eqref{eq:max_prob_of_erasure}.
Thus, for any sequence of \mbox{blocklength-$n$} \mbox{rate-$R$} codebooks with maximal
probability of erasure approaching zero, 
\begin{equation}
R \leq \lambda.
\end{equation}
A standard expurgation argument shows that this is also true when we 
replace the maximal probability of erasure with the average (over the messages)
probability of erasure. 
\end{proof}

\section{Graph-Theoretic Preliminaries}
\label{sec:graphs}
A directed graph (or simply a graph) $G$ is described by its finite vertex
set $V(G)$
and its edge set $E(G)\subset V(G)\times V(G)$.
We say that there is an edge from $x$ to $y$ in $G$ if $(x,y) \in E(G)$.
We always assume that $G$ does not contain self-loops, i.e., that $(x,x) \notin
E(G)$ for all~$x\in V(G)$.

The strong product of two graphs $G$ and $H$ is denoted by $G\times H$;
its vertex set is $V(G)\times V(H)$, and there is an edge
from $(x,y)$ to $(x',y')$ in $G \times H$ if either $(x,x') \in E(G)$ and $(y,y') \in E(H)$, 
or if $(x,x') \in E(G)$ and~$y=y'$, or if~$x=x'$ and~$(y,y') \in E(H)$. The
$n$-fold strong product of $G$ with itself is denoted by $G^n$. 

The subgraph of $G$ induced by $A\subseteq V(G)$ is the graph whose 
vertex set is~$A$ and whose edge set is $E(G)\cap (A\times A)$.

A subset $A\subseteq V(G)$ is an independent set in~$G$ if the subgraph
of $G$ it induces has no edges, i.e., if $E(G) \cap (A\times A) =
\emptyset$. The maximum cardinality of an independent set in $G$ is denoted by
$\alpha(G)$. 
We define the Sperner capacity of $G$ as
\begin{equation}
\label{eq:sperner}
\Sigma(G) = \lim_{n\to\infty} \frac{1}{n} \log \alpha(G^n),
\end{equation}
where the limit on the RHS is equal to the supremum because 
the sequence $\alpha(G^{1}), \alpha(G^{2}),\ldots$ is supermultiplicative.\footnote{A sequence
$a_1,a_2,\ldots$ of real numbers is supermultiplicative if $a_{n+m} \geq a_n
a_m$ for all $m$ and $n$.}
The reader is warned that this definition is not standard. However, it is
equivalent to the original definition given in~\cite{gargano1992qualitative}; see
Appendix~\ref{appendix:sperner_def}.

A path in $G$ is a sequence of $n\geq 2$ distinct vertices $x_1,\ldots,x_n$ such
that $(x_j,x_{j+1}) \in E(G)$ for all $j\in \{1,\ldots,n-1\}$. The first vertex
in this path is $x_1$, and the last vertex is $x_n$. We say that there is a
path from~$x$ to~$y$ in~$G$ if there is a path in~$G$ whose first vertex is
$x$ and whose last vertex is $y$.

A cycle is a path $x_1,\ldots,x_n$ with $(x_n,x_1) \in E(G)$. 
(Note that with this definition a bidirectional edge is a cycle.)
We say that $G$ is acyclic if it does not contain a cycle. 
The maximum cardinality of a subset $A\subseteq V(G)$ that induces an 
acyclic subgraph of $G$ is denoted by $\rho(G)$. 

The following two results will be key in the proof of
Theorem~\ref{thm:upper_bound}. The first is that~$\alpha$ can 
be replaced with~$\rho$ in~\eqref{eq:sperner}.
\begin{theorem}
\label{thm:alex}
For every graph $G$, 
\begin{equation}
\label{eq:alex}
\Sigma(G) = \lim_{n\to\infty} \frac{1}{n} \log \rho(G^{n}),
\end{equation}
and the limit is equal to the supremum.
\end{theorem}
In particular, Theorem~\ref{thm:alex} asserts that
\begin{equation}
\label{eq:alex_cor}
\rho(G^{n}) \leq e^{n \Sigma(G)},\quad\text{for all $n$.}
\end{equation}
A proof of Theorem~\ref{thm:alex} is provided in Appendix~\ref{appendix:a}.
(An anonymous reviewer pointed out that a statement equivalent to 
Theorem~\ref{thm:alex} is apparently well-known; see
Appendix~\ref{appendix:sperner_def}.)

The number of edges of $G$ ending in a vertex $x$ is called the
in-degree of $x$ in $G$ and is denoted by $\indeg(x,G)$, i.e., 
\begin{equation}
\indeg(x,G) = \bigl|\bigl\{x' \in V(G): (x',x) \in E(G) \bigr\}\bigr|.
\end{equation}
The next result is a slight generalization of
\cite[p. 95, Theorem 1]{alon2011probabilistic}.
\begin{theorem}
\label{thm:caro}
For every graph $G$, 
\begin{equation}
\sum_{x\in V(G)} \frac{1}{1+\indeg(x,G)} \leq \rho(G).
\end{equation}
\end{theorem}
A proof of Theorem~\ref{thm:caro} is provided in Appendix~\ref{appendix:b}.

For DMCs $W$ with $\set{X}\subseteq \set{Y}$ and $W(x|x)>0$ for 
every $x\in \set{X}$, we define the associated graph~$G(W)$ 
to have vertex set $\set{X}$ and edge set comprising all 
ordered pairs~$(x,y)$ of distinct elements of $\set{X}$ for which $W(y|x)>0$. 
Thus, for such channels we have
\begin{equation}
\label{eq:sperner_equal_detection}
\Csp(W) = \Sigma\bigl(G(W)\bigr).
\end{equation}
Indeed, every Sperner code for $W$ of blocklength $n$ is an independent set in
$G(W)^n$ and vice versa. 
We note the identity 
\begin{equation}
\label{eq:induced_graph_power}
G(W^n) = G(W)^{n}.
\end{equation}

\section{Proof of Theorem~\ref{thm:upper_bound}}
\label{sec:proof}
Applying Jensen's Inequality to the RHS of~\eqref{eq:multiletter} yields
\begin{equation}
\Ceo \leq \sup_{n\geq 1} \frac{1}{n} \max_{P\in U_n} \log
\sum_{\bfy\in \supp(PW^n)} \frac{(PW^n)(\bfy)}{P(\set{X}^n(\bfy))}.
\end{equation}
It thus suffices to show that for all $P\in U_n$,
\begin{equation}
\label{eq:305}
\sum_{\bfy \in \supp(PW^n)} \frac{(PW^n)(\bfy)}{P(\set{X}^n(\bfy))} \leq
\bigl(e^{\Csp} + \varepsilon \card{\set{X}}(\card{\set{Y}}-1)\bigr)^n.
\end{equation}
Fix then some $P \in U_n$. Since the labels do not matter, we may assume for simplicity of notation that
$\set{X}=\{0,\ldots,\card{\set{X}}-1\}$ and
$\set{Y}=\{0,\ldots,\card{\set{Y}}-1\}$, where $\card{\set{Y}} \geq
\card{\set{X}}$. The distribution on~$\set{Y}^n$ induced by~$P$ and~$W^n$ can be written as 
\begin{equation}
\label{eq:271}
(PW^n)(\bfy) = \sum_{\substack{\bfz\in \set{Y}^n:\\\bfy+\bfz\in \set{X}^n}} P(\bfy+ \bfz)
W^n(\bfy|\bfy+ \bfz),
\end{equation}
where addition is to be understood component-wise modulo~$\card{\set{Y}}$. 
The $\varepsilon$-noise property \eqref{eq:low_noise} implies
\begin{equation}
\label{eq:279}
W^n(\bfy|\bfy+ \bfz) \leq \varepsilon^{\zeronorm{\bfz}},\quad \text{if $\bfy+
\bfz \in \set{X}^n$},
\end{equation}
where $\zeronorm{\bfz}$ denotes the number of nonzero components of $\bfz$. 
Thus, starting with the LHS of~\eqref{eq:305},
\begin{align}
&\sum_{\bfy \in \supp(PW^n)} \frac{(PW^n)(\bfy)}{P(\set{X}^n(\bfy))}\\
&= \sum_{\bfy \in \supp(PW^n)} 
\sum_{\substack{\bfz\in \set{Y}^n:\\\bfy+\bfz\in \set{X}^n}} \frac{P(\bfy+ \bfz)
W^n(\bfy|\bfy+ \bfz)}{P(\set{X}^n(\bfy))}\label{eq:379}\\
&= \sum_{\bfz \in \set{Y}^n} \sum_{\substack{\bfy \in \set{Y}^n:\\\bfy+ \bfz \in
\set{X}^n\\P(\bfy+ \bfz)>0\\W^n(\bfy|\bfy+ \bfz)>0}} \frac{P(\bfy+ \bfz)
W^n(\bfy|\bfy+ \bfz)}{P(\set{X}^n(\bfy))}\label{eq:382}\\
&\leq \sum_{\bfz \in \set{Y}^n} \varepsilon^{\zeronorm{\bfz}}\!\!\!\!\!\!\!\! 
\sum_{\substack{\bfy \in
\set{Y}^n:\\\bfy+ \bfz \in
\set{X}^n\\P(\bfy+ \bfz)>0\\W^n(\bfy|\bfy+ \bfz)>0}} \frac{P(\bfy+
\bfz)}{P(\set{X}^n(\bfy))}\label{eq:385}\\
&=\sum_{\bfz \in \set{Y}^n} \varepsilon^{\zeronorm{\bfz}}\!\!\!\!\!\!\!\!
\sum_{\substack{\bfy \in \set{X}^n:\\ P(\bfy)>0\\W^n(\bfy-\bfz|\bfy)>0}}
\frac{1}{\card{\{\bfx\in \supp(P): W^n(\bfy-\bfz|\bfx)>0\}}},\label{eq:389}
\end{align}
where~\eqref{eq:379} follows from~\eqref{eq:271}; where~\eqref{eq:382} follows
by changing the order of summation and dropping terms that are zero; where~\eqref{eq:385} follows
from~\eqref{eq:279}; and where~\eqref{eq:389} follows by substituting $\bfy$ for
$\bfy+\bfz$ and because $P$ is uniform
over its support. 
For every $\bfz\in \set{Y}^n$, let $P_\bfz$ be any PMF on $\set{X}^n$ of support
\begin{equation}
\supp(P_\bfz)=\{\bfx \in \set{X}^n: P(\bfx)W^n(\bfx-\bfz|\bfx)>0 \}.
\end{equation}
(In fact, $P_\bfz$ could be any nonnegative function with the above support.)
Also define for every~$\bfz \in \set{Y}^n$ the channel 
\begin{equation}
W_\bfz(\bfy|\bfx) = W^n(\bfy-\bfz|\bfx),
\end{equation}
with input alphabet $\supp(P_\bfz)$ and output alphabet $\set{Y}^n$. 
Since $\supp(P_\bfz)
\subseteq \supp(P)$, 
\begin{multline}
\label{eq:some_ineq_810}
\card{\{\bfx\in \supp(P): W^n(\bfy-\bfz|\bfx)>0\}}\\ \geq \card{\{\bfx\in
\supp(P_\bfz): W_\bfz(\bfy|\bfx)>0\}}.
\end{multline}
Using~\eqref{eq:some_ineq_810} we can 
upper-bound the inner sum on the RHS of \eqref{eq:389} by 
\begin{equation}
\label{eq:315}
\sum_{\bfy \in \supp(P_\bfz)} \frac{1}{\card{\{\bfx\in \supp(P_\bfz):
W_\bfz(\bfy|\bfx)>0\}}}.
\end{equation}
This sum can also be written as
\begin{equation}
\label{eq:324}
\sum_{\bfy \in V(G(W_\bfz))} \frac{1}{1+\indeg(\bfy,G(W_\bfz))},
\end{equation}
where $G(W_\bfz)$ is the graph associated with the channel $W_\bfz$ (see
Section~\ref{sec:graphs}). Since~\eqref{eq:324} is upper-bounded 
by~$\rho(G(W_\bfz))$ (Theorem~\ref{thm:caro}), we thus have
\begin{equation}
\label{eq:830_7}
\sum_{\bfy \in \supp(PW^n)} \frac{(PW^n)(\bfy)}{P(\set{X}^n(\bfy))} \leq \sum_{\bfz
\in \set{Y}^n} \varepsilon^{\zeronorm{\bfz}} \rho\bigl(G(W_\bfz)\bigr).
\end{equation}
We next argue that
\begin{equation}
\label{eq:331}
\rho\bigl(G(W_\bfz)\bigr) \leq
\card{\set{X}}^{\zeronorm{\bfz}}\rho\bigl(G(W)^{n-\zeronorm{\bfz}}\bigr),
\end{equation}
where we define $\rho(G(W)^{0}) =1$. 
When~$\zeronorm{\bfz}=n$, then~\eqref{eq:331} is trivial, so we assume that
$0\leq \zeronorm{\bfz} < n$. Let~$\bfx(\bfz)$ denote the restriction of $\bfx\in \set{X}^n$ to the nonzero components of
$\bfz$, and let $\bfx(\comp{\bfz})$ denote the restriction of~$\bfx$ to the
zero components of~$\bfz$. We will prove~\eqref{eq:331} by contradiction. In
order to reach a contradiction, assume
that for some integer~$\eta$ strictly larger than the RHS of~\eqref{eq:331} there exist 
distinct vertices~$\bfx_1,\ldots,\bfx_{\eta}$ in $\supp(P_\bfz)$ that induce an acyclic
subgraph of~$G(W_\bfz)$. 
Partition this collection of vertices by placing into the
same class all $\bfx_j$'s that have the same restriction $\bfx_j(\bfz)$. Since there
are $\card{\set{X}}^{\zeronorm{\bfz}}$ such classes, one of them must contain 
$\kappa>\rho(G(W)^{n-\zeronorm{\bfz}})$ vertices; call them $\bfx'_1,\ldots,\bfx'_{\kappa}$. 
Since $\bfx'_1,\ldots,\bfx'_{\kappa}$ are distinct, and since their restrictions
to the nonzero components of $\bfz$ are identical, their restrictions to the
zero components of $\bfz$, i.e.,
$\bfx'_1(\comp{\bfz}),\ldots,\bfx'_{\kappa}(\comp{\bfz})$ must all be distinct. 
Also, if $\bfx,\bfy \in \supp(P_\bfz)$ and $\bfx(\bfz) = \bfy(\bfz)$, then 
\begin{equation}
W_\bfz(\bfy|\bfx)>0 \iff
W^{n-\zeronorm{\bfz}}\bigl(\bfy(\comp{\bfz})\big|\bfx(\comp{\bfz})\bigr)>0.
\end{equation}
It follows that the subgraph of $G(W_\bfz)$ induced by
$\bfx'_1,\ldots,\bfx'_{\kappa}$ is isomorphic to the subgraph of
$G(W^{n-\zeronorm{\bfz}})$ induced by
$\bfx'_1(\comp{\bfz}),\ldots,\bfx'_{\kappa}(\comp{\bfz})$.\footnote{The
isomorphism is $\bfx \mapsto \bfx(\comp{\bfz})$.} 
And since the former
is acyclic, so must the latter be, which is a contradiction because 
$G(W^{n-\zeronorm{\bfz}}) = G(W)^{n-\zeronorm{\bfz}}$ and 
$\kappa > \rho(G(W)^{n-\zeronorm{\bfz}})$. 

Having established~\eqref{eq:331}, we further note that by~\eqref{eq:alex_cor}
and~\eqref{eq:sperner_equal_detection},
\begin{equation}
\label{eq:912_3432}
\rho\bigl(G(W)^{n-\zeronorm{\bfz}}\bigr) \leq e^{(n-\zeronorm{\bfz}) \Csp}.
\end{equation}

By combining~\eqref{eq:830_7}, \eqref{eq:331}, and~\eqref{eq:912_3432}, we
obtain
\begin{align}
&\sum_{\bfy \in \supp(PW^n)} \frac{(PW^n)(\bfy)}{P(\set{X}^n(\bfy))}\notag\\
&\quad\leq 
\sum_{\bfz \in \set{Y}^n} \varepsilon^{\zeronorm{\bfz}}
\card{\set{X}}^{\zeronorm{\bfz}}
e^{(n-\zeronorm{\bfz}) \Csp}\label{eq:766_342}\\
&\quad=\sum_{k=0}^n \binom{n}{k} (\card{\set{Y}}-1)^k \varepsilon^k \card{\set{X}}^k
e^{(n-k)\Csp},\label{eq:768_342}
\end{align}
where the equality follows because the summand on the RHS of~\eqref{eq:766_342} depends on $\bfz$ only via
$\zeronorm{\bfz}$ and there are~$\binom{n}{k}(\card{\set{Y}}-1)^k$ elements in $\set{Y}^n$
with exactly $k$ nonzero components. 
This completes the proof because the RHS of~\eqref{eq:768_342} 
is equal to the RHS of~\eqref{eq:305}.\qed

\section{Concluding Remarks}
\label{sec:remarks}
\begin{enumerate}
\item In Theorem~\ref{thm:upper_bound} we may replace $\card{\set{Y}}$ 
with $\card{\set{X}}+2^{\card{\set{X}}}-1$. Indeed, using Proposition~\ref{prop:combine} and 
noting that the $\varepsilon$-noise property and~$\Csp$ are preserved if we
combine only output symbols in $\set{Y}\setminus \set{X}$, we can reduce the 
output alphabet to at most $\card{\set{X}} +2^{\card{\set{X}}}-1$ symbols.

\item The Sperner capacity of a graph and the Sperner capacity of a channel are
of course just different formulations of the same problem. Indeed, in
Section~\ref{sec:graphs} we noted that
$\Csp(W) = \Sigma(G(W))$, where $G(W)$ is the graph associated with the
$\varepsilon$-noise channel~$W$. Conversely, 
we may associate with each directed graph $G$ a canonical
$\varepsilon$-noise channel $W_\varepsilon(G)$ by choosing $\set{X}=\set{Y}=V(G)$ and
\begin{equation}
W_\varepsilon(G)(y|x) = \begin{cases} 1-\varepsilon& y=x,\\ \frac{\varepsilon}{\outdeg(x,G)}& (x,y) \in
E(G),\\0&\text{otherwise,}\end{cases}
\end{equation}
if $\outdeg(x,G)\geq 1$, and otherwise
\begin{equation}
W_\varepsilon(G)(y|x) =\begin{cases} 1&\text{if $y=x$,}\\0 &\text{if $y\neq
x$,}\end{cases}
\end{equation}
where $\outdeg(x,G)$ denotes the out-degree of $x$ in $G$, i.e., the
number of edges of $G$ emanating from the vertex $x$. Then for any $\varepsilon
\in (0,1)$, 
\begin{equation}
\Sigma(G) = \Csp(W_\varepsilon(G)). 
\end{equation}

\item We mentioned in the introduction that $\varliminf_{\varepsilon\to 0} \Ceo
\geq \Csp$ was proved in \cite{ahlswede1996erasure}.
Here, we offer a simple proof based on the $n$-letter version of Forney's
lower bound~\eqref{eq:multi_forney}. Given $\delta>0$ choose $n$ 
so that there exists a Sperner code of size~$e^{n(\Csp-\delta)}$.
Let $P$ be the uniform PMF on the codebook and note that
\begin{equation}
\label{eq:531}
(PW^n)(\bfx) \geq P(\bfx)(1-\varepsilon)^n\quad\text{for all $\bfx \in \set{X}^n$.} 
\end{equation}
For this $P$, 
\begin{align}
\Ceo &\geq \frac{1}{n} \sum_{\bfy \in \set{Y}^n} (PW^n)(\bfy) \log
\frac{1}{P(\set{X}^n(\bfy))}\label{eq:536}\\
&\geq \frac{1}{n} \sum_{m=1}^{e^{n(\Csp-\delta)}} (PW^n)(\bfx_m) \log
\frac{1}{P(\set{X}^n(\bfx_m))}\label{eq:538}\\
&= \frac{1}{n} \sum_{m=1}^{e^{n(\Csp-\delta)}} (PW^n)(\bfx_m) \log
\frac{1}{P(\bfx_m)}\label{eq:539}\\
&\geq (1-\varepsilon)^n\frac{1}{n} \sum_{m=1}^{e^{n(\Csp-\delta)}}
P(\bfx_m) \log \frac{1}{P(\bfx_m)}\label{eq:540}\\
&=(1-\varepsilon)^n(\Csp-\delta),\label{eq:541}
\end{align}
where~\eqref{eq:536} follows from~\eqref{eq:multi_forney}; 
where~\eqref{eq:538} follows because the summand is nonnegative; 
where~\eqref{eq:539} follows from the choice of $P$ and the property of Sperner codes
\eqref{eq:zero_error_detection}; 
where~\eqref{eq:540} follows from~\eqref{eq:531}; and where~\eqref{eq:541}
follows from the choice of $P$. 
Letting first~$\varepsilon \to 0$ and then~$\delta \to 0$ completes
the proof.
\item For some channels the bound in
Theorem~\ref{thm:upper_bound} can be sharpened. For example, for
the canonical $\varepsilon$-noise channel associated
with the cyclic orientation of a triangle $\Csp = \log
2$~\cite{calderbank1993sperner,blokhuis1993sperner}, and it is shown
in \cite{bunte2013zero} that $\Ceo \leq \log 2$ for every $\varepsilon \in (0,1)$. 
\end{enumerate}

\appendices
\section{Proof of Theorem~\ref{thm:alex}}
\label{appendix:a}
We shall need the elementary fact that the vertices of any acyclic graph $G$
can be labeled with the numbers~$1,\ldots,\card{V(G)}$ such that $(x,y) \in E(G)$ only if $x<y$ (see, e.g.,
\cite[Section~5.7]{thulasiraman2011graphs}).\footnote{A different way to state
this is that any partial order on a finite set can be extended to a total order
on this set.}

Using this fact, we first show that the sequence $\rho(G^{1}),\rho(G^{2}),\ldots$ is supermultiplicative,
which will imply that the limit on the RHS of~\eqref{eq:alex} equals the supremum. 
Choose for each~$n$ some $A_n \subseteq V(G)^n$ that achieves~$\rho(G^{n})$, i.e.,
$A_n$ induces an acyclic subgraph of $G^{n}$ and $\card{A_n} =
\rho(G^{n})$. We show that $A_n \times A_m$ induces an acyclic subgraph of
$G^{n+m}$ and hence that
\begin{align}
\rho\bigl(G^{n+m}\bigr) &\geq \card{A_n \times A_m}\\
&= \rho(G^{n}) \rho(G^{m}).\label{eq:supermult}
\end{align}
Label the vertices in $A_n$ with the numbers $1,\ldots,\card{A_n}$ 
so that $(x,x') \in E(G^n) \cap (A_n\times A_n)$
implies $x < x'$. Similarly label the vertices in $A_m$. 
To reach a contradiction, assume that $(x_1,y_1),\ldots,(x_\eta,y_\eta)$ is a
cycle in the subgraph of $G^{n+m}$ induced by $A_n \times A_m$. 
From the definition of strong product and the labeling of the vertices it
follows that $x_1 < x_\eta$ or $y_1< y_\eta$. Consequently, there cannot be an edge from $(x_\eta,y_\eta)$ to $(x_1,y_1)$ in this
subgraph, which contradicts the assumption that
$(x_1,y_1),\ldots,(x_\eta,y_\eta)$ is a cycle. 


As to~\eqref{eq:alex}, we first show that
\begin{equation}
\label{eq:acyclic_186}
\Sigma(G) = \log \card{V(G)},\quad\text{for all acyclic $G$.}
\end{equation}
Note that this will prove Theorem~\ref{thm:alex} in the special case where $G$
is acyclic. Indeed, in this case $\rho(G)= \card{V(G)}$, so~\eqref{eq:supermult} 
implies $\rho(G^n) \geq \card{V(G)}^n$. And since clearly $\rho(G^n) \leq
\card{V(G)}^n$, we thus have 
\begin{equation}
\rho(G^n) = \card{V(G)}^n,\quad\text{for all acyclic $G$.} 
\end{equation}
To prove~\eqref{eq:acyclic_186}, note that $\alpha(G^n) \leq
\card{V(G)}^n$ and hence $\Sigma(G) \leq \log \card{V(G)}$ (this is
true for any~$G$, not just acyclic), so it only remains to
prove the reverse inequality. 
Since~$G$ is acyclic, we may label its vertices with the numbers
$1,\ldots,\card{V(G)}$ so that there is an edge 
from $x$ to $y$ in $G$ only if $x<y$. We then define the weight 
of a vertex~$\bfx$ in~$G^{n}$ as the sum of the labels of its $n$ components. 
Thus, the weight is a number between~$n$ and~$n \card{V(G)}$. 

As we next show, if $A$ is a subset of $V(G)^n$ all of whose members have the same weight, 
then~$A$ is an independent set in $G^{n}$. Indeed, if
$\bfx$ and $\bfy$ are distinct vertices in $A$, then~$x^{(j)} > y^{(j)}$, say, 
for some~$j\in\{1,\ldots,n\}$. Since $\bfx$ and $\bfy$ have equal weight, 
there must also be some~$k\neq j$ for which~$x^{(k)}<y^{(k)}$.
Thus, $(x^{(j)},y^{(j)})\notin E(G)$ 
and $(y^{(k)},x^{(k)})\notin E(G)$, 
so there is no edge from $\bfx$ to $\bfy$ and no edge from $\bfy$ to
$\bfx$ in $G^{n}$. 

If we partition $V(G)^n$ by putting in the same class all
vertices of the same weight, then one of the $n\card{V(G)}-n+1$ different 
classes must have at least
\begin{equation}
\frac{\card{V(G)}^n}{n\card{V(G)}-n+1} 
\end{equation}
members. Thus, 
\begin{equation}
\frac{1}{n} \log \alpha(G^{n}) \geq \log \card{V(G)} - \frac{1}{n} \log\bigl( n
\card{V(G)} - n +1\bigr),
\end{equation}
and letting $n$ tend to infinity establishes $\Sigma(G) \geq \log \card{V(G)}$
and hence proves~\eqref{eq:acyclic_186}.

To complete the proof of~\eqref{eq:alex}, let $G$ be any graph (not necessarily
acyclic) and let $\lambda$ denote the limit of~$n^{-1} \log \rho(G^{n})$ as $n$
tends to infinity (i.e., the supremum).  
For a given~$\delta>0$ select $\nu$ so that 
\begin{equation}
\label{eq:1059_32}
\frac{1}{\nu} \log \rho(G^{\nu}) \geq \lambda - \delta.
\end{equation}
Choose $A\subseteq V(G)^\nu$ that achieves $\rho(G^{\nu})$ and let $H$
denote the acyclic subgraph of $G^{\nu}$ it induces. Since $H^{m}$ is the subgraph of $G^{\nu m}$
induced by $A^m$, 
\begin{equation}
\frac{1}{\nu m} \log \alpha(G^{\nu m}) \geq \frac{1}{\nu m} \log
\alpha(H^{m}).
\end{equation}
Letting $m$ tend to infinity, we obtain
\begin{equation}
\Sigma(G) \geq \frac{1}{\nu} \Sigma(H).\label{eq:1071_342}
\end{equation}
Since $H$ is acyclic, we can substitute it for $G$ in~\eqref{eq:acyclic_186} to
obtain
\begin{align}
\frac{1}{\nu} \Sigma(H)&= \frac{1}{\nu} \log \card{A}\\
&=\frac{1}{\nu} \log \rho\bigl(G^{\nu}\bigr)\label{eq:1077_23},
\end{align}
where \eqref{eq:1077_23} follows because $A$ achieves $\rho(G^{\nu})$.
Combining~\eqref{eq:1071_342}, \eqref{eq:1077_23}, and~\eqref{eq:1059_32} 
shows that $\Sigma(G) \geq \lambda - \delta$. Since this is true for every $\delta>0$, 
\begin{equation}
\Sigma(G) \geq \lambda.
\end{equation}
On the other hand, a graph with no edges is trivially acyclic, so $\alpha(G^{n}) \leq
\rho(G^{n})$ and hence~$\Sigma(G) \leq \lambda$.
\qed

\section{Proof of Theorem~\ref{thm:caro}}
\label{appendix:b}
Let $<$ be a total ordering of the vertices of $G$ 
and consider the subset~$A\subseteq V(G)$ comprising all $x\in V(G)$ such that if
$(x',x) \in E(G)$ for some $x'\in V(G)$, then $x'<x$. The subgraph of~$G$ 
induced by $A$ is acyclic. Indeed, if~$x_1,\ldots,x_\eta$ is a path in this subgraph, then
$x_1<x_\eta$, so we cannot have $(x_\eta,x_1) \in E(G)$. Thus, 
\begin{equation}
\label{eq:1100_824}
\card{A} \leq \rho(G).
\end{equation}
Suppose now that $<$ is drawn uniformly at random among all total orderings of
$V(G)$. Then
\begin{equation}
\label{eq:prob_equal_indegree}
\Pr(x \in A) = \frac{1}{1+\indeg(x,G)},\quad \text{for all $x\in V(G)$.}
\end{equation}
Indeed, $x$ is in $A$ if, and only if, it is the greatest vertex in the
set
\begin{equation}
B=\{x\}\cup \{x': (x',x) \in E(G)\}.
\end{equation}
Since $<$ is drawn uniformly at random, every vertex in $B$ has the same probability of being the greatest element in
$B$, so~\eqref{eq:prob_equal_indegree} follows by noting that $\card{B}=1+\indeg(x,G)$.

Summing both sides of~\eqref{eq:prob_equal_indegree} over all vertices of $G$ yields
\begin{equation}
\label{eq:1006_3432}
\sum_{x\in V(G)} \frac{1}{1+\indeg(x,G)} = \sum_{x\in V(G)} \Pr(x\in A).
\end{equation}
By writing $\Pr(x\in A)$ as the expectation of the indicator function of the event
$\{x\in A\}$ and by swapping summation and expectation, we see that the RHS
of~\eqref{eq:1006_3432} is the
expected cardinality of~$A$. This expected cardinality cannot exceed $\rho(G)$
because~\eqref{eq:1100_824} holds for every realization of~$<$.
\qed

\section{Remarks about the Definition of Sperner Capacity}
\label{appendix:sperner_def}
It was pointed out in Section~\ref{sec:graphs} that the definition of Sperner
capacity used in this paper \eqref{eq:sperner} is not standard. 
Here we compare this definition to the original one given in
\cite{gargano1992qualitative} and discuss their
equivalence. But first, we need some more notation.

The weak product of two directed graphs $G$ and $H$ is denoted by $G \cdot H$.
The vertex set of $G\cdot H$ is $V(G) \times V(H)$ and there is an edge from
$(x,y)$ to $(x',y')$ in $G\cdot H$ if $(x,x') \in
E(G)$ or $(y,y')\in E(H)$. The $n$-fold weak product of $G$ with itself is
denoted by $G^{\vee n}$. 

The complement $G^c$ of a directed graph $G$ has vertex set $V(G^c)=V(G)$ 
and edge set 
\begin{equation*}
E(G^c) = \bigl(V(G) \times V(G)\bigr) \setminus 
\bigl(E(G) \cup \{(x,x): x\in V(G)\}\bigr), 
\end{equation*}
i.e., $E(G^c)$ is 
the set complement of $E(G)$ in $V(G) \times V(G)$ minus the self-loops. Note
that $(G^c)^c = G$. 

A subset $A\subseteq V(G)$ is a \emph{symmetric clique} in $G$ 
if the subgraph it induces 
is complete, i.e., if $(x,y) \in E(G)$ whenever $x$ and $y$ are distinct vertices
in $A$. The largest cardinality of a symmetric clique in $G$ is denoted 
by~$\omega(G)$. In \cite{gargano1992qualitative}, 
the Sperner capacity of $G$ is defined as
\begin{equation}
\label{eq:sperner_def_standard}
\Sigma_0(G) = \lim_{n\to\infty} \frac{1}{n} \log \omega(G^{\vee n}).
\end{equation}
(We use $\Sigma_0$ instead of $\Sigma$ to distinguish it 
from our nonstandard definition~\eqref{eq:sperner}.)
This definition is equivalent to the definition~\eqref{eq:sperner} in the
following sense: For every directed graph $G$, 
\begin{equation}
\label{eq:equiv_sperner1093}
\Sigma(G) = \Sigma_0(G^c).
\end{equation}
To prove this relationship, 
first note that an independent set in~$G$ is a symmetric clique 
in~$G^c$ and vice versa. Hence $\alpha(G) = \omega(G^c)$. Moreover, one easily
verifies that $(G^n)^c = (G^c)^{\vee n}$. Consequently, 
\begin{align}
\alpha(G^n) &= \omega((G^n)^c)\\
&=\omega((G^c)^{\vee n}),
\end{align}
and~\eqref{eq:equiv_sperner1093} follows.

The reason we prefer~\eqref{eq:sperner} over~\eqref{eq:sperner_def_standard}
is Theorems~\ref{thm:caro} and~\ref{thm:alex}.

An anonymous reviewer pointed out to us that 
an equivalent version of Theorem~\ref{thm:alex} is well-known. 
To state it, we need the concept of \emph{transitive cliques}:

A subset of vertices $A\subseteq V(G)$ is called a transitive clique
in $G$ if there exists a total order $<$ on $A$ with the property
\begin{equation}
(\text{$x,y\in A$ and $x<y$}) \implies (x,y) \in E(G).
\end{equation}
It is easily verified that a subset $A\subseteq V(G)$ induces 
an acyclic subgraph of $G$ if, and
only if, $A$ is a transitive clique in $G^c$.
Thus, if $\omega_{\textnormal{tr}}(G)$ denotes the size of the largest
transitive clique in $G$, then 
\begin{equation}
\rho(G)=\omega_{\textnormal{tr}}(G^c).
\end{equation}
Hence, Theorem~\ref{thm:alex} is equivalent to the statement that
\begin{equation}
\Sigma_0(G) = \lim_{n\to\infty} \frac{1}{n}\log \omega_{\textnormal{tr}}(G^{\vee n}),
\end{equation}
i.e., that $\omega$ can be replaced with $\omega_{\textnormal{tr}}$ in the
original definition of Sperner capacity~\eqref{eq:sperner_def_standard}.
Although we could not find a reference where this is explicitly proven, it was
certainly used in~\cite{korner1998clique}.

\section*{Acknowledgment}

We are grateful for valuable comments and suggestions 
from the anonymous reviewers and the Associate Editor. 


\ifCLASSOPTIONcaptionsoff
  \newpage
\fi


\begin{thebibliography}{10}
\providecommand{\url}[1]{#1}
\csname url@samestyle\endcsname
\providecommand{\newblock}{\relax}
\providecommand{\bibinfo}[2]{#2}
\providecommand{\BIBentrySTDinterwordspacing}{\spaceskip=0pt\relax}
\providecommand{\BIBentryALTinterwordstretchfactor}{4}
\providecommand{\BIBentryALTinterwordspacing}{\spaceskip=\fontdimen2\font plus
\BIBentryALTinterwordstretchfactor\fontdimen3\font minus
  \fontdimen4\font\relax}
\providecommand{\BIBforeignlanguage}[2]{{%
\expandafter\ifx\csname l@#1\endcsname\relax
\typeout{** WARNING: IEEEtran.bst: No hyphenation pattern has been}%
\typeout{** loaded for the language `#1'. Using the pattern for}%
\typeout{** the default language instead.}%
\else
\language=\csname l@#1\endcsname
\fi
#2}}
\providecommand{\BIBdecl}{\relax}
\BIBdecl

\bibitem{csiszar1995channel}
I.~Csisz\'ar and P.~Narayan, ``Channel capacity for a given decoding metric,''
  \emph{{IEEE} Trans. Inf. Theory}, vol.~41, no.~1, pp. 35--43, 1995.

\bibitem{ahlswede1996erasure}
R.~Ahlswede, N.~Cai, and Z.~Zhang, ``Erasure, list, and detection zero-error
  capacities for low noise and a relation to identification,'' \emph{{IEEE}
  Trans. Inf. Theory}, vol.~42, no.~1, pp. 55--62, 1996.

\bibitem{forney1968exponential}
G.~Forney~Jr, ``Exponential error bounds for erasure, list, and decision
  feedback schemes,'' \emph{{IEEE} Trans. Inf. Theory}, vol.~14, no.~2, pp.
  206--220, 1968.

\bibitem{telatar1992phd}
{\.I}.~E. Telatar, ``Multi-access communications with decision feedback
  decoding,'' Ph.D. dissertation, Massachusetts Institute of Technology, Dept.
  of Electrical Engineering and Computer Science, May 1992.

\bibitem{pinsker1970transmission}
M.~S. Pinsker and A.~Y. Sheverdyaev, ``Transmission capacity with zero error
  and erasure,'' \emph{Problemy Peredachi Informatsii}, vol.~6, no.~1, pp.
  20--24, 1970.

\bibitem{hui1983phd}
J.~Y.~N. Hui, ``Fundamental issues of multiple accessing,'' Ph.D. dissertation,
  Massachusetts Institute of Technology, Dept. of Electrical Engineering and
  Computer Science, November 1983.

\bibitem{csiszar1981graph}
I.~Csisz{\'a}r and J.~K{\"o}rner, ``Graph decomposition: A new key to coding
  theorems,'' \emph{{IEEE} Trans. Inf. Theory}, vol.~27, no.~1, pp. 5--12,
  1981.

\bibitem{balakirsky1995converse}
V.~B. Balakirsky, ``A converse coding theorem for mismatched decoding at the
  output of binary-input memoryless channels,'' \emph{{IEEE} Trans. Inf.
  Theory}, vol.~41, no.~6, pp. 1889--1902, 1995.

\bibitem{bunte2012zero}
C.~Bunte and A.~Lapidoth, ``The zero-undetected-error capacity of discrete
  memoryless channels with feedback,'' in \emph{Communication, Control, and
  Computing (Allerton), 2012 50th Annual Allerton Conference on}.\hskip 1em
  plus 0.5em minus 0.4em\relax IEEE, 2012, pp. 1838--1842.

\bibitem{nakiboglu2012errors}
B.~Nakibo\u{g}lu and L.~Zheng, ``Errors-and-erasures decoding for block codes
  with feedback,'' \emph{{IEEE} Trans. Inf. Theory}, vol.~58, no.~1, pp.
  24--49, 2012.

\bibitem{telatar1997zero}
{\.I}.~E. Telatar, ``Zero-error list capacities of discrete memoryless
  channels,'' \emph{{IEEE} Trans. Inf. Theory}, vol.~43, no.~6, pp. 1977--1982,
  1997.

\bibitem{bunte2013listsize}
C.~{Bunte} and A.~{Lapidoth}, ``On the listsize capacity with feedback,''
  \emph{submitted to the IEEE Transactions on Information Theory}, Nov. 2013,
  \url{http://arxiv.org/abs/1311.0195}.

\bibitem{bunte2013cutoff}
C.~Bunte and A.~Lapidoth, ``On the cutoff rate and the average-listsize
  capacity of discrete memoryless channels with feedback,'' in
  \emph{Information Theory Workshop (ITW), 2013 IEEE}, 2013, pp. 649--653.

\bibitem{gallager1968information}
R.~G. Gallager, \emph{{Information Theory and Reliable Communication}}.\hskip
  1em plus 0.5em minus 0.4em\relax New York: John Wiley \& Sons, 1968.

\bibitem{arikan1996inequality}
E.~Ar{\i}kan, ``An inequality on guessing and its application to sequential
  decoding,'' \emph{{IEEE} Trans. Inf. Theory}, vol.~42, no.~1, pp. 99--105,
  1996.

\bibitem{polya1998problems}
G.~P{\'o}lya and G.~Szeg{\H{o}}, \emph{Problems and Theorems in Analysis
  I}.\hskip 1em plus 0.5em minus 0.4em\relax Berlin Heidelberg:
  Springer-Verlag, 1978.

\bibitem{gargano1992qualitative}
L.~Gargano, J.~K{\"o}rner, and U.~Vaccaro, ``Qualitative independence and
  {S}perner problems for directed graphs,'' \emph{Journal of Combinatorial
  Theory, Series A}, vol.~61, no.~2, pp. 173--192, 1992.

\bibitem{alon2011probabilistic}
N.~Alon and J.~H. Spencer, \emph{The Probabilistic Method}, 3rd~ed.\hskip 1em
  plus 0.5em minus 0.4em\relax Hoboken, NJ: Wiley, 2008.

\bibitem{calderbank1993sperner}
A.~R. Calderbank, P.~Frankl, R.~L. Graham, W.-C.~W. Li, and L.~A. Shepp, ``The
  {S}perner capacity of linear and nonlinear codes for the cyclic triangle,''
  \emph{Journal of Algebraic Combinatorics}, vol.~2, no.~1, pp. 31--48, 1993.

\bibitem{blokhuis1993sperner}
A.~Blokhuis, ``On the {S}perner capacity of the cyclic triangle,''
  \emph{Journal of Algebraic Combinatorics}, vol.~2, no.~2, pp. 123--124, 1993.

\bibitem{bunte2013zero}
C.~Bunte, A.~Lapidoth, and A.~Samorodnitsky, ``The zero-undetected-error
  capacity of the low-noise cyclic triangle channel,'' in \emph{Information
  Theory Proceedings (ISIT), 2013 IEEE International Symposium on}, 2013, pp.
  91--95.

\bibitem{thulasiraman2011graphs}
K.~Thulasiraman and M.~N.~S. Swamy, \emph{Graphs: Theory and Algorithms}.\hskip
  1em plus 0.5em minus 0.4em\relax New York: John Wiley \& Sons, 1992.

\bibitem{korner1998clique}
J.~K{\"o}rner, ``On clique growth in products of directed graphs,''
  \emph{Graphs and Combinatorics}, vol.~14, no.~1, pp. 25--36, 1998.

\end{thebibliography}
\end{document}